\documentclass[letterpaper, 10 pt, conference]{ieeeconf}  

\IEEEoverridecommandlockouts                              

\title{\LARGE \bf
Synchronous Observer Design for Landmark-Inertial SLAM with Magnetometer and Intermittent GNSS Measurements
}

\author{Arkadeep Saha$^{1}$, Pieter van Goor$^{2}$ and Ravi Banavar$^{1}$
\thanks{$^{1}$Centre for Systems and Control, Indian Institute of Technology Bombay, Powai, Mumbai-400076, India.
        {\tt\small 22b1270@iitb.ac.in, banavar@iitb.ac.in}}%
\thanks{$^{2}$School of Aerospace, Mechanical and Mechatronic Engineering (AMME), Faculty of Engineering, University of Sydney, NSW, 2006, Australia.
        {\tt\small  pieter.vangoor@sydney.edu.au}}%
}

\newcommand{\evec}{\mathbf{e}_3}

\newcommand{\SO}{\mathbf{SO}(3)}
\newcommand{\soalg}{\mathfrak{so}(3)}

\newcommand{\SEn}{\mathbf{SE}_{n+2}(3)}
\newcommand{\senalg}{\mathfrak{se}_{n+2}(3)}

\newcommand{\SIMn}{\mathbf{SIM}_{n+2}(3)}
\newcommand{\simnalg}{\mathfrak{sim}_{n+2}(3)}
\newcommand{\diag}{\text{diag}}
\newcommand{\tr}{\text{tr}}
\newcommand{\cset}[2]{\left\{ #1 \; \middle| \; #2 \right\}}

\usepackage{graphicx}          
\usepackage{amsmath}
\usepackage{amssymb}
\usepackage{amsfonts}
\usepackage{tikz-cd}
\usepackage{color}

\newtheorem{lemma}{Lemma}
\newtheorem{thm}{Theorem}
\newtheorem{definition}{Definition}

\begin{document}

\maketitle
\thispagestyle{empty}
\pagestyle{empty}

\begin{abstract}

In Landmark-Inertial Simultaneous Localisation and Mapping (LI-SLAM), the positions of landmarks in the environment and the robot's pose relative to these landmarks are estimated using landmark position measurements, and measurements from the Inertial Measurement Unit (IMU).
However, the robot and landmark positions in the inertial frame, and the yaw of the robot, are not observable in LI-SLAM.
This paper proposes a nonlinear observer for LI-SLAM that overcomes the observability constraints with the addition of intermittent GNSS position and magnetometer measurements.
The full-state error dynamics of the proposed observer is shown to be both almost-globally asymptotically stable and locally exponentially stable, and this is validated using simulations.

\end{abstract}

\section{INTRODUCTION}

Simultaneous Localisation and Mapping (SLAM) is the problem of concurrently estimating  a robot's pose and a map of its environment.
It has been an active area of research in mobile robotics since last thirty years \cite{bailey_slam:part2}.
The two standard approaches in SLAM are the extended Kalman filter (EKF) \cite{bailey_slam:part2} and graph-based nonlinear optimization, \cite{kaess2012isam2}, both of which offer only local convergence, but recent research in the nonlinear observer community has generated novel solutions \cite{van2021constructive, boughellaba2025nonlinear} that guarantee (almost-)global stability.

LI-SLAM is the version of the SLAM problem that uses measurements of angular velocity and acceleration from an IMU alongside complementary exteroceptive landmark position measurements provided by a stereo camera, or an RGB-D camera, or a lidar.
This problem has been studied by the nonlinear observer community using a variety of Kalman filters.
A globally asymptotically stable Kalman filter for LI-SLAM expressed in the body-frame of the robot, and a Procrustes problem to estimate the robot's pose in the inertial frame was proposed in \cite{lourencco2016simultaneous}.
The LI-SLAM problem was addressed by \cite{johansen2016globally} using an attitude heading reference system (AHRS) to estimate the robot's attitude, and a linear time-varying Kalman filter for the remaining mapping and position estimation problem.
The invariant EKF (IEKF) for landmark SLAM (where the robot's body-frame velocity measurement is assumed to be available) was introduced in \cite{barrau2015ekf} using a novel Lie group $\mathbf{SE}_{n+1}(3)$.
Each of these approaches face limitations with all of them lacking global stability guarantee and requiring at least a quadratically scaling computational complexity associated with the Kalman filter.
Deterministic nonlinear observers for SLAM follow from a rich history of geometric observers for attitude estimation \cite{mahony2008nonlinear} and pose estimation \cite{baldwin2009nonlinear} using Lie groups.
In \cite{mahony2017geometric}, a geometric observer for the kinematic landmark SLAM was designed by introducing the $\mathbf{SLAM}_n(3)$ Lie group and defining a quotient manifold structure to encode the invariance of SLAM to changes in the inertial reference frame that led to inconsistency issues in classic approaches.
A nonlinear observer on the $\mathbf{SE}_{n+3}(3)$ Lie group for LI-SLAM with almost-global asymptotic stability using gravity direction as an additional auxiliary state in the observer has been proposed in \cite{boughellaba2025nonlinear}.
A synchronous observer design for LI-SLAM was proposed in \cite{saha2025synchronous} inspired by the recent developments in synchronous observers for velocity-aided attitude (VAA) in \cite{van2023vaa}, and inertial navigation systems in \cite{VANGOOR2025112328}. 
This work uniquely identified the observable states to a base space and proved the stability properties in the base space.

It is well-studied that the problem of LI-SLAM is not fully observable. 
The observable states in LI-SLAM are: roll and pitch of the robot's attitude, body-frame velocity of the robot, and the relative positions of the landmarks in the robot's frame \cite{martinelli2013observability}.
Thus, the yaw and the positions of the robot and landmarks in any inertial frame are not observable.
But for many practical applications, such as outdoor navigation in urban environments, it is necessary to estimate the full attitude and positions with respect to the desired inertial frame.
Additional measurements, such as GNSS position and magnetometer readings can be used to overcome this observability constraint of the LI-SLAM.
The GNSS position measurement has been incorporated in SLAM by both EKF \cite{kim20056dof} and nonlinear optimization-based approaches \cite{boche2025okvis2}.

This paper proposes a nonlinear observer for LI-SLAM aided with additional measurements from the magnetometer and GNSS using the principles of synchronous observer design.
The intermittency in the GNSS measurements has been incorporated, and the constraints 
arising out of the duration of its availability, have been accounted for in the convergence of the proposed observer.
Although the GNSS position measurement is sufficient for full observability of LI-SLAM, additional magnetometer measurements are used for faster convergence of yaw.

This paper consists of four sections alongside the introduction and the conclusion. 
Section \ref{sec:preliminaries} introduces the mathematical preliminaries and notations used in the paper. 
Section \ref{sec:prob_description} provides the description and Lie group interpretation of the LI-SLAM problem. 
In Section \ref{sec:obs_design}, we provide the observer architecture, the design of correction terms, and the proofs of stability and convergence. 
The simulation results are provided in Section \ref{sec:simulations}, verifying the theory developed throughout the paper. 

\section{Preliminaries}\label{sec:preliminaries}

For all matrices $A,B \in \mathbb{R}^{n\times m}$, the Frobenius/Euclidean inner product and norm are defined by 
\begin{align*}
    \langle A,B\rangle = \text{tr}(A^\top B), && |A|^2 = \text{tr} (A^\top A),
\end{align*}
respectively. 
For a positive definite matrix $P\in \mathbb{R}^{m\times m}$ and any matrix $A\in \mathbb{R}^{n\times m}$, define the weighted norm
\begin{align*}
    |A|_P^2 = \langle A, AP\rangle = \text{tr}(APA^\top).
\end{align*}
The 2-sphere is defined by $\mathbf{S}^2 = \cset{y \in \mathbb{R}^3}{|y| = 1}$.
$\mathbf{1}_n\in \mathbb{R}^n$ and $\mathbf{0}_n\in \mathbb{R}^n$ are column vectors with all 1's and 0's, respectively. \

\begin{definition}[Temporally Persistently Exciting] \label{def:TPE}
A switching signal $\sigma : [0,\infty) \to \{0,1\}$ is said to be \emph{temporally persistently exciting} (TPE) if
there exist constants $T > \tau > 0$ such that, for every $t > 0$ there exist an interval $[t_1,t_2) \subset [t,t+T)$ with $t_2 - t_1 \geq \tau$ and $\sigma(s) = 1$ for all $s \in  [t_1,t_2)$.
\end{definition}

Thus, a signal $\sigma : [0,\infty) \to \{0,1\}$ is TPE if, within every interval of length $T$, $\sigma$ is constantly equal to 1 over a subinterval of length at least $\tau$. 
This is key to our proof of convergence in Theorem~\ref{thm:main_result}.

\subsection{Lie groups}

For an introduction to Lie groups and smooth manifolds, the authors recommend \cite{2012_lee_IntroductionSmoothManifolds}.
The special orthogonal group $\SO$ is the Lie group of 3D rotations, defined as
\begin{align*}
    \SO := \cset{R\in \mathbb{R}^{3\times 3}}{R^\top R = I_3, \det(R) = 1}.
\end{align*}
For any vector $\Omega \in \mathbb{R}^3$, define $\Omega^\times \in \mathbb{R}^{3\times 3}$ as
\begin{align*}
    \Omega^\times = \begin{pmatrix}
        0 & -\Omega_3 & \Omega_2 \\
        \Omega_3 & 0 & -\Omega_1 \\
        -\Omega_2 & \Omega_1 & 0
    \end{pmatrix}.
\end{align*}
Then $\Omega^\times v = \Omega\times v$ for any vector $v \in \mathbb{R}^3$, where $\times$ is the usual cross product.
The Lie algebra of $\SO$ is defined as 
\begin{align*}
    \soalg := \cset{\Omega^\times \in \mathbb{R}^{3 \times 3}}{\Omega \in \mathbb{R}^3}.
\end{align*}
The extended special Euclidean group $\mathbf{SE}_n(3)$ and its Lie algebra $\mathfrak{se}_n(3)$ are defined by \cite{barrau2015ekf}
\begin{align*}
    \mathbf{SE}_n(3) &:=\left\{ \begin{pmatrix}
        R & V\\
        0_{n\times3} & I_n
    \end{pmatrix} \middle|\ R \in \mathbf{SO}(3),\ V\in \mathbb{R}^{3\times n} \right\},\\
    \mathfrak{se}_n(3) &:=\left\{ \begin{pmatrix}
        \Omega^\times & W\\
        0_{n\times3} & 0_{n\times n}
    \end{pmatrix} \middle|\ \Omega \in \mathbb{R}^3,\ W\in \mathbb{R}^{3\times n} \right\}.
\end{align*}
An element of $\mathbf{SE}_n(3)$ may be denoted $X = (R, V)$ for convenience, where $R \in \mathbf{SO}(3)$ and $V\in \mathbb{R}^{3\times n}$. Likewise, an element of $\mathfrak{se}_n(3)$ can be denoted by $\Delta = (\Omega_\Delta, W_\Delta)$, where $\Omega_\Delta \in \mathbb{R}^3$ and $W_\Delta \in \mathbb{R}^{3\times n}$. The matrix Lie group $\mathbf{SIM}_n(3)$ and its Lie algebra $\mathfrak{sim}_n(3)$ are defined as 
\begin{align*}
    &\mathbf{SIM}_n(3) := \\ &\left\{ \begin{pmatrix}
        R & V\\
        0_{n\times3} & A
    \end{pmatrix} \Bigg|\ R \in \mathbf{SO}(3),\ V\in \mathbb{R}^{3\times n}, A\in \mathbf{GL}(n) \right\},\\
    &\mathfrak{sim}_n(3) :=\\ &\left\{ \begin{pmatrix}
        \Omega^\times & W\\
        0_{n\times3} & S
    \end{pmatrix} \Bigg|\ \Omega \in \mathbb{R}^3,\ W\in \mathbb{R}^{3\times n}, \ S \in \mathfrak{gl}(n) \right\},
\end{align*}
where $\mathbf{GL}(n)$ is the set of $n\times n$ invertible matrices.
An element of $\mathbf{SIM}_n(3)$ can be denoted $Z = (R_Z, V_Z, A_Z)$ for convenience, where $R_Z\in \mathbf{SO}(3)$, $V_Z\in \mathbb{R}^{3\times n}$ and $A_Z \in \mathbf{GL}(n)$.
Likewise, an element of $\mathfrak{se}_n(3)$ can be denoted by $\Gamma = (\Omega_\Gamma, W_\Gamma, S_\Gamma)$, where $\Omega_\Gamma \in \mathbb{R}^3$, $W_\Gamma \in \mathbb{R}^{3\times n}$ and $S_\Gamma \in \mathfrak{gl}(n)$.

\section{Problem Description}\label{sec:prob_description}

\subsection{LI-SLAM Dynamics and Measurements}

Consider a mobile robot equipped with an IMU, a \mbox{3-D} camera system, a GNSS receiver, and a magnetometer moving in an environment with $n$ landmarks.
Let $\{0\}$ be the inertial frame as defined by the GNSS receiver.
Let $\{B\}$ be the body-fixed frame of the robot. 
We assume that the frames of the IMU , the camera system and the magnetometer are coincident and align with $\{B\}$, for simplicity. 
Define $R \in \SO$ as the orientation of $\{B\}$ with respect to $\{0\}$ expressed in $\{0\}$.
Define $v \in \mathbb{R}^3$ and $x \in \mathbb{R}^3$ as the velocity and position of the robot in the inertial frame $\{0\}$. 
Define the positions of the $n$ landmarks in the inertial frame $\{0\}$ as $p_i \in \mathbb{R}^3$ for $ i = 1, \cdots, n$.

The angular velocity and proper acceleration of the robot in the body frame obtained from the IMU are denoted by $\Omega \in \mathbb{R}^3$ and $a \in \mathbb{R}^3$.
The \emph{LI-SLAM state} is given by $(R,v,x,p_1, \cdots, p_n)$ and the \emph{LI-SLAM inputs} are $(\Omega, a)$. 
Then, the dynamics of the LI-SLAM system are
\begin{align}\label{eq:dyn}
   \dot{R} = R\Omega^\times, \quad
   \dot{v} = Ra + g\mathbf{e}_3, \quad
   \dot{x}  = v, \quad
   \dot{p}_i  = 0,
\end{align}
where $i = 1, \cdots,n$. 
The measurements from the 3-D camera system are the landmark positions relative to $\{B\}$.
Thus, the measurement of landmark $i$ is given by
\begin{align}\label{eq:landmark_meas_model}
    y_i = R^\top(p_i- x) \in \mathbb{R}^3.
\end{align}
The magnetometer gives the magnetic field direction measured in the body frame, that is
\begin{align}\label{eq:compass_meas_model}
    y_m = R^\top \mathring{y}_m \in \mathbf{S}^2,
\end{align}
where $\mathring{y}_m\in \mathbf{S}^2$ is the constant reference magnetic field direction.
The GNSS receiver provides measurements of the robot position in the inertial frame $\{0\}$.
However, we consider that GNSS measurements may not always be available, so the measurement model for the GNSS position is
\begin{align}\label{eq:gnss_meas_model}
    y_x = \sigma x \in \mathbb{R}^3, 
\end{align}
where signal $\sigma(t) \in \{0,1\}$ represents the availability of GNSS measurements, that is, $\sigma(t) = 1$ when measurements are available; otherwise $\sigma(t) = 0$. 
We assume that $\sigma$ is TPE according to definition \ref{def:TPE}, with time constants $T>\tau>0$.

\subsection{Lie Group Interpretation}
The LI-SLAM state can be written using $\mathbf{SE}_{n+2}(3)$. 
Define the state $X$ by
\begin{equation}\label{eq:Liegrp}
\begin{aligned}
    X &= \begin{pmatrix}
        R & V \\
        0_{(n+2) \times 3} & I_{n+2}
    \end{pmatrix} \in \SEn, \\
    V &=\begin{pmatrix}
        v & x & p_1 & \cdots & p_n
    \end{pmatrix} \in \mathbb{R}^{3 \times (n+2)},
\end{aligned}
\end{equation}
where $V$ is the translational sub-matrix containing the velocity and position of the robot and the landmark positions.
Following \cite{saha2025synchronous}, the system dynamics \eqref{eq:dyn} can be rewritten as dynamics on the Lie group as
\begin{align}\label{eq:lie_grp_dynamics}
    \dot{X} &= XU + GX + NX - XN, \\
    U & = \begin{pmatrix}
        \Omega^\times & W_U\\
        0_{(n+3)\times 3} & 0_{(n+2)\times (n+2)}
    \end{pmatrix},\ W_U = \begin{pmatrix}
        a & 0_{3 \times (n+1)}
    \end{pmatrix}, \notag \\
    G &= \begin{pmatrix}
        0_{3 \times 3} & W_G\\
        0_{(n+2)\times 3} & 0_{(n+2)\times (n+2)}
    \end{pmatrix},\ W_G = \begin{pmatrix}
        g\evec & 0_{3 \times (n+1)}
    \end{pmatrix},\notag \\
    N &= \begin{pmatrix}
        0_{3 \times 3} & 0_{3\times (n+2)}\\
        0_{(n+2)\times 3} & S_N
    \end{pmatrix},\ S_N = \begin{pmatrix}
        0 & -1 & \mathbf{0}_n^\top \\
        0 & 0 & \mathbf{0}_n^\top \\
        \mathbf{0}_n & \mathbf{0}_n & 0_{n\times n}
    \end{pmatrix}.\notag
\end{align}
Here, $U \in \senalg$ is the matrix of the IMU inputs, $G \in \senalg$ is the constant matrix for the gravity vector, and $N \in \simnalg$ is the constant matrix that represents the velocity dynamics of the robot position.
The system dynamics are similar to the dynamics of inertial navigation system studied in \cite{VANGOOR2025112328}.

The measurement model for the landmarks \eqref{eq:landmark_meas_model} can also be written with Lie group notation as 
\begin{align}\label{eq:landmark_meas_matrix}
    Y_p &= \begin{pmatrix}
        y_1 & \cdots & y_n
    \end{pmatrix} \in \mathbb{R}^{3 \times n}, \\
    \begin{pmatrix}
        Y_p \\ C
    \end{pmatrix} &= X^{-1}\begin{pmatrix}
        0_{3\times n} \\ C
    \end{pmatrix}= \begin{pmatrix}
        -R^\top V C  \\
        C
    \end{pmatrix}, 
\end{align}
where 
\begin{align*}
    C = \begin{pmatrix}
        \mathbf{0}_n^\top \\
        \mathbf{1}_n^\top\\
        -I_n
    \end{pmatrix}\in \mathbb{R}^{(n+2)\times n}.
\end{align*}
Similarly, the measurement model for the GNSS position \eqref{eq:gnss_meas_model} can also be written as 
\begin{align}
    \begin{pmatrix}
        y_x \\ \sigma C_x
    \end{pmatrix} =  X\begin{pmatrix}
        0_{3\times 1}\\
        \sigma C_x
    \end{pmatrix} = \begin{pmatrix}
        \sigma VC_x \\
        \sigma C_x
    \end{pmatrix}, 
\end{align}
where 
\begin{align*}
    C_x = \begin{pmatrix}
        0 \\ 1 \\ \mathbf{0}_n
    \end{pmatrix}\in \mathbb{R}^{(n+2)\times 1}.
\end{align*}
In summary, we have shown that the system state, dynamics and measurements admit powerful and compact matrix representation using the Lie group structure of $\SEn$. 
For the remainder of the paper, we identify the state space with the Lie group $ \SEn$.

\section{Observer Design}\label{sec:obs_design}
\subsection{Synchronous Observer Architecture}
We follow the synchronous observer architecture developed for group-affine systems in \cite{van2021autonomous, VANGOOR2025112328}.
Let the system state be $X \in \SEn$ as in \eqref{eq:Liegrp}.
The observer state is defined to include the state estimate $\hat{X}\in \SEn$ and an auxiliary state $Z \in \SIMn$, with the dynamics
\begin{align}\label{eq:obs_dyn}
    \dot{\hat{X}} &= \hat{X}U + G\hat{X}+N\hat{X} - \hat{X}N + Z\Delta Z^{-1}\hat{X}, \notag \\
    \dot{Z} &= (G+N)Z - Z\Gamma,
\end{align}
where $\Delta\in \senalg$ and $\Gamma\in \simnalg$ are the Lie algebra-valued correction terms to be designed later.
The associated observer error is defined as 
\begin{align}
    \bar{E} := Z^{-1}X\hat{X}^{-1}Z \in \SEn,
\end{align}
and its dynamics are given by \cite{van2021autonomous}
\begin{align}\label{eq:obs_error}
    \dot{\bar{E}} = \Gamma \bar{E} - \bar{E}\Gamma - \bar{E}\Delta.
\end{align}
Note that the error dynamics are independent of the system input and the state estimate, and they only depend linearly on the correction terms $\Delta$ and $\Gamma$. 
Hence, the observer and the system are \emph{$\bar{E}$-synchronous} \cite{van2021autonomous}.

For simplicity in the observer design and analysis, we choose to fix the rotation term in $Z$ to $R_Z \equiv I_3$ by setting $\Omega_\Gamma = 0_{3\times 3}$.
Thus, in the following sections, the components of $Z$ are $(R_Z, V_Z, A_Z) \equiv (I_3, V_Z, A_Z)$, with only $V_Z$ and $A_Z$ being dynamic.
The rotational and translational components of the error $R_{\bar{E}} \in \SO$ and $V_{\bar{E}}\in \mathbb{R}^{3\times (n+2)}$ are obtained as
\begin{align*}
    R_{\bar{E}} &= R\hat{R}^\top, \\
    V_{\bar{E}} &= (VA_Z - V_Z) - R_{\bar{E}}(\hat{V}A_Z - V_Z).
\end{align*}
Their dynamics are given by
\begin{subequations}
\begin{align}
    \dot{R}_{\bar{E}} &= -R_{\bar{E}}\Omega_\Delta^\times, \label{eq:error_dyn_RE} \\
    \dot{V}_{\bar{E}} &= -V_{\bar{E}} S_\Gamma + (I_3 - R_{\bar{E}})W_\Gamma - R_{\bar{E}}W_\Delta.  \label{eq:error_dyn_VE}
\end{align}
\end{subequations}
Our goal is to design the correction terms $(\Omega_\Delta, W_\Delta)$ and $(W_\Gamma, S_\Gamma)$ such that the error $\bar{E} \to I$ asymptotically.

\subsection{Design of Correction Terms}

By taking advantage of the synrhony of the observer error, we may design individual correction terms for the GNSS position, the landmark positions, and the magnetometer measurements such that they independently decrease a chosen Lyapunov function candidate.
The final correction term is then chosen as the sum of the individual correction terms for each sensor.
As the error dynamics depend linearly on the correction terms, it follows that the Lyapunov function derivative will thus be negative semi-definite for the final correction term, as shown in Theorem 5.4 in \cite{VANGOOR2025112328}.
We propose the Lyapunov function $\mathcal{L}(\bar{E})$ for the observer error $\bar{E} = (R_{\bar{E}}, V_{\bar{E}})\in \SEn$ as
\begin{align}\label{eq:lyap_fn}
    \mathcal{L}(\bar{E}) := |V_{\bar{E}}|^2 + \tr(I_3 - R_{\bar{E}}).
\end{align}
For arbitrary correction terms $\Delta \in \senalg, \Gamma \in \simnalg$, the derivative of this Lyapunov function along the error dynamics \eqref{eq:error_dyn_RE}, \eqref{eq:error_dyn_VE} is
\begin{align}
    \dot{\mathcal{L}}(\Delta, \Gamma) &= \langle V_{\bar{E}}, -V_{\bar{E}} S_\Gamma + (I_3 - R_{\bar{E}})W_\Gamma - R_{\bar{E}}W_\Delta \rangle \notag \\& \qquad + \tr(R_{\bar{E}} \Omega_\Delta^\times).
\end{align}

\subsubsection{GNSS Position Correction Terms}

We label the correction terms for the GNSS position as $\Delta^x = (\Omega_\Delta^x, W_\Delta^x) \in \senalg$ and $\Gamma^x = (0_{3\times 3}, W_\Gamma^x, S_\Gamma^x) \in \simnalg$. 

\begin{lemma}\label{lem:gnss_correction}
Choose the positive gains $k_x, k_{Rx}, q > 0$ and define the correction terms for the GNSS position measurements as
\begin{align*}
    W_\Delta^x &= (k_x+ k_{Rx})(y_x - \sigma \hat{x})C_x^\top A_Z^{-\top}, \\
    W_\Gamma^x &= -(k_x + k_{Rx})(y_x - \sigma V_ZA_Z^{-1}C_x)C_x^\top A_Z^{-\top}, \\
    \Omega_\Delta^x &= 4k_{Rx} \sigma(\hat{x} - V_ZA_Z^{-1}C_x)^\times (y_x - \sigma V_ZA_Z^{-1}C_x), \\
    S_\Gamma^x &= -\frac{k_x\sigma}{2}A_Z^{-1}C_xC_x^\top A_Z^{-\top} + qI_{n+2}.
\end{align*}
Then, the derivative of the Lyapunov function satisfies 
\begin{align*}
    \dot{\mathcal{L}}(\Delta^x, \Gamma^x) &\le -2k_{Rx}(|(I-R_{\bar{E}}^2)(y_x - \sigma V_ZA_Z^{-1}C_x)| \notag \\
    & \qquad - \sigma |V_{\bar{E}}A_Z^{-1}C_x|)^2-\sigma k_x|V_{\bar{E}}A_Z^{-1}C_x|^2 \notag\\ &\qquad - 2q|V_{\bar{E}}|^2 \le 0 
\end{align*}
\end{lemma}

\begin{proof}
    The proof closely follows \cite[Lemma 5.5, 5.6]{VANGOOR2025112328}, with minor modifications to include $\sigma$ and simplify $S_\Gamma^x$.
\end{proof}
\subsubsection{Correction terms using landmark positions}
Let the correction terms using the landmark position measurements be $\Delta^p = (\Omega_\Delta^p, W_\Delta^p) \in \senalg$ and $\Gamma^p = (0_{3\times 3}, W_\Gamma^p, S_\Gamma^p) \in \simnalg$.
\begin{lemma}\label{lem:landmark_correction}
Choose the positive gains $k_p, k_{Rp} > 0$ and define the landmark position measurements correction terms as
\begin{align*}
    W_\Delta^p &= -(k_p + nk_{Rp})\hat{R}(Y - \hat{Y})C^\top A_Z^{-\top}, \\
    W_\Gamma^p &= (k_p + nk_{Rp})V_ZA_Z^{-1}CC^\top A_Z^{-\top}, \\
    S_\Gamma^p &= -\frac{k_p}{2}A_Z^{-1}CC^\top A_Z^{-\top}, \\
    \Omega_\Delta^p &= 4k_{Rp} (V_ZA_Z^{-1}C\mathbf{1}_n)^\times \hat{R}(Y-\hat{Y})\mathbf{1}_n,
\end{align*}
where $Y$ is the landmark position measurement in \eqref{eq:landmark_meas_matrix}, and $\hat{Y} = -\hat{R}^\top \hat{V}C$ is the estimated landmark position.
Then, the derivative of the Lyapunov function satisfies 
\begin{align*}
    &\dot{\mathcal{L}}(\Delta^p, \Gamma^p) \\&\le  - 2k_{Rp}(\sqrt{n} |V_{\bar{E}}A_Z^{-1}C| - | (I- R_{\bar{E}}^2)V_ZA_Z^{-1}C \mathbf{1}_n |)^2 \\
    &\qquad -k_p|V_{\bar{E}}A_Z^{-1}C|^2 \le 0
\end{align*}
\end{lemma}
\begin{proof}
For the complete proof, refer to Appendix~\ref{app:proof_lemma_landmark}.
\end{proof}
\subsubsection{Correction terms using magnetometer}
Let the correction terms designed using the magnetometer direction readings be $\Delta^m = (\Omega_\Delta^m, W_\Delta^m)\in \senalg$ and $\Gamma^m = (0_{3\times 3}, W_\Gamma^m, S_\Gamma^m)\in \simnalg$.
\begin{lemma}\label{lem:magneto_correction}
    Choose the positive gain $k_m>0$ and define the correction terms for magnetometer measurement as
    \begin{align*}
        W_\Delta^m &=0_{3\times(n+2)}, \;\; W_\Gamma^m = 0_{3\times(n+2)}, \;\; S_\Gamma^m = 0_{(n+2)\times(n+2)},\\
        \Omega_\Delta^m &= 4k_m (\hat{R}y_m) \times\mathring{y}_m.
    \end{align*}
 Then, the derivative of the Lyapunov function is
    \begin{align*}
        \dot{\mathcal{L}}(\Delta^m, \Gamma^m) = -2k_m |(I-R_{\bar{E}})^2\mathring{y}_m|^2 \le 0
    \end{align*}
\end{lemma}
\begin{proof}
    See \cite[Lemma 5.8]{VANGOOR2025112328}.
\end{proof}

\subsection{Observer Analysis}

Having defined correction terms for each individual measurement type, we combine them in the following theorem and show that the Lyapunov function is stable and that the observer error converges almost-globally asymptotically to the identity.

\begin{thm}\label{thm:main_result}
    Consider the LI-SLAM system state $X \in \SEn$ with dynamics \eqref{eq:lie_grp_dynamics} and the measurements in \mbox{(\ref{eq:landmark_meas_model}-\ref{eq:gnss_meas_model})}. 
    Let $\hat{X} \in \SEn$ and $Z \in \SIMn$ be the state estimate and the auxiliary state of the observer, respectively, with dynamics \eqref{eq:obs_dyn}.
    Define the correction terms $\Delta = \sum_{i=\{x,p,m\}}\Delta^i$ and $\Gamma = \sum_{i=\{x,p,m\}}\Gamma^i$, where $(\Delta^x, \Gamma^x)$, $(\Delta^p, \Gamma^p)$ and $(\Delta^m, \Gamma^m)$ are the correction terms designed in lemma \ref{lem:gnss_correction}, \ref{lem:landmark_correction} and \ref{lem:magneto_correction}, respectively.
    Choose the observer gains $k_x, k_p > 0$ and $q >0$ satisfying 
    \begin{align}\label{eq:condition_gain}
          2n\tau q e^{-2qT}k_p + (8q^2\tau^2 e^{-4qT}-1)k_x > 0.
    \end{align}
    Initialize $A_Z(0)$ such that $P_0 := A_Z(0)A_Z(0)^\top$ satisfies
    \begin{equation}\label{eq:P_initial_conditions}
    \begin{aligned}
        &P_0 = \begin{pmatrix}
        s_v(0) & s_{vx}(0) & \frac{k_p}{4q^2}\mathbf{1}_n^\top \\
        s_{vx}(0) & s_{x}(0) & -\frac{k_p}{2q} \mathbf{1}_n^\top \\
        \frac{k_p}{4q^2}\mathbf{1}_n & -\frac{k_p}{2q} \mathbf{1}_n & \frac{k_p}{2q} I_n
    \end{pmatrix}, \\
    &\frac{nk_p}{2q} + k_x e^{-2qT}\tau \le s_x(0) \le \frac{nk_p+k_x}{2q}, \\
    &-\frac{nk_p+k_x}{4q^2} \le s_{vx}(0) \le -\frac{nk_p}{4q^2} -\frac{ k_x e^{-2qT}\tau}{2q}, \\
    &\frac{nk_p}{4q^3} +\frac{ k_x e^{-2qT}\tau}{2q^2}\le s_v(0) \le \frac{nk_p+k_x}{4q^3}.
    \end{aligned}
    \end{equation}
    Let the observer error be $\bar{E}$ as defined in \eqref{eq:obs_error}, then
    \begin{itemize}
        \item[(i)] The correction terms are bounded for all time.
        \item[(ii)] The translation error $V_{\bar{E}}$ is globally exponentially stable to zero.
        \item[(iii)] The rotation error $R_{\bar{E}}$ is almost globally asymptotically stable (AGAS) and locally exponentially stable to the identity.
        Hence, $\bar{E}$ is AGAS to the identity with the unstable equilibria given by 
        \begin{align*}
            \mathcal{E}_u = \cset{\bar{E}\in \SEn}{\tr(R_{\bar{E}})=-1}.
        \end{align*}
        \item[(iv)] If the error converges to identity, $\bar{E} \to I_{n+5}$, then the estimated state converges to the true state, $\hat{X}\to X$.
    \end{itemize}
\end{thm}
\begin{proof}
\underline{Proof of item (i)}: To show that the correction terms are bounded, we begin by showing that the varying elements of the auxiliary state, $A_Z$ and $V_Z$ are bounded.
Define $P := A_ZA_Z^\top$.
We begin by showing that $P$ is lower and upper bounded.
The dynamics of $P$ are given by
\begin{align*}
    \dot{P} &= \dot{A}_Z A_Z^\top + A_Z \dot{A}_Z^\top \\
    &= S_NP + PS_N^\top - A_Z(S_\Gamma + S_\Gamma^\top) A_Z^\top \\
    &= S_NP + P S_N^\top + \sigma k_x C_xC_x^\top + k_p CC^\top - 2qP \\
    &= (S_N - qI)P + P(S_N - qI)^\top + \sigma k_x C_xC_x^\top + k_p CC^\top
\end{align*}
Under the initial conditions specified in \eqref{eq:P_initial_conditions} and the assumption that $\sigma$ is temporally persistently exciting, it can be shown that the eigenvalues of the matrix $P$ are always bounded above and below (see Appendix~\ref{app:proof_P_bounded}).
Therefore, $A_Z$ is also lower and upper bounded by non-zero constants.
The dynamics of $V_Z$ have the form
\begin{align}
    \dot{V}_Z = -V_Z M + B, 
\end{align}
where the matrices $M$ and $B$ are defined in appendix \ref{app:V_z_dynamics}.
Importantly, $B$ is bounded, and $M$ is a negative definite matrix with $M < -q I$.
Therefore,
\begin{align*}
    \frac{\mathrm{d}}{\mathrm{d} t} \vert V_Z \vert^2 
     \le 2\langle V_Z, B \rangle - 2q |V_Z|^2
    \le |V_Z||B| - q |V_Z|^2,
\end{align*}
and thus $V_Z$ is upper bounded.
The correction terms are the sum and product of the auxiliary state, state estimate and measurements, all of which are bounded, so the correction terms are themselves also bounded.

\underline{Proof of item (ii)}: 
Following \eqref{eq:error_dyn_VE}, the dynamics of the translation error $V_{\bar{E}}$ are 
\begin{align*}
    \dot{V}_{\bar{E}} &= - \sigma \left(\frac{k_x}{2} + k_{Rx}\right)V_{\bar{E}} A_Z^{-1}C_xC_x^\top A_Z^{-\top} - qV_{\bar{E}}\\
    &\qquad -\left(\frac{k_p}{2} + nk_{Rp}\right) V_{\bar{E}}A_Z^{-1}CC^\top A_Z^{-\top}.
\end{align*}
Define the cost function $\mathcal{L}_V = |V_{\bar{E}}|^2$.
Then, the dynamics of the cost $\mathcal{L}_V$ is given by 
\begin{align*}
    \dot{\mathcal{L}}_V &= -\sigma(k_x + 2k_{Rx})|V_{\bar{E}}A_Z^{-1}C_x|^2 \\ &\quad-(k_p + 2nk_{Rp}) |V_{\bar{E}} A_Z^{-1}C|^2 - 2q|V_{\bar{E}}|^2 \\
    & \le -2q |V_{\bar{E}}|^2 \le -2q \mathcal{L}_V.
\end{align*}
Therefore, the translation error $V_{\bar{E}}$ is globally exponentially stable and converges to zero.

\underline{Proof of item (iii)}:
We consider the Lyapunov function defined in \eqref{eq:lyap_fn}. 
Due to modularity, the derivative of the Lyapunov function is simply the sum of the derivatives in case of individual corrections, hence
\begin{align}
    \dot{\mathcal{L}} 
    &\le  -2k_{Rx}\sigma(|(I-R_{\bar{E}}^2)\mu_x| - |V_{\bar{E}}A_Z^{-1}C_x|)^2  \notag \\
    & \qquad - 2k_{Rp}(| (I- R_{\bar{E}}^2)\mu_Z | - \sqrt{n} |V_{\bar{E}}A_Z^{-1}C|)^2 \notag\\
    & \qquad -\sigma k_x|V_{\bar{E}}A_Z^{-1}C_x|^2 - 2q|V_{\bar{E}}|^2  -k_p|V_{\bar{E}}A_Z^{-1}C|^2 \notag\\
    & \qquad -2k_m |(I-R_{\bar{E}})^2\mathring{y}_m|^2, \label{eq:Ldot}
\end{align}
where $\mu_x = (x -  V_ZA_Z^{-1}C_x)$ and $\mu_Z = V_ZA_Z^{-1}C \mathbf{1}_n$.
This is clearly negative semi-definite, but it is notably not continuous, let alone uniformly continuous, due to the presence of the signal $\sigma$.
Nonetheless, as $\dot{\mathcal{L}}\le 0$ and $\mathcal{L}\ge0$ by construction, we have $0\le \mathcal{L}(\infty)\le \mathcal{L}(t) \leq \mathcal{L}(0)$. 
Therefore, the integral $\int_0^\infty \dot{\mathcal{L}}(t)dt = \mathcal{L}(\infty) - \mathcal{L}(0)$ is bounded and it must be that $\dot{\mathcal{L}} \to 0$.
Moreover, since all the summands in \eqref{eq:Ldot} are non-positive, it must be that each individual term converges to zero.
Combining this with the fact that $V_{\bar{E}} \to 0 $ globally exponentially, then
\begin{align*}
     -2k_{Rx} \sigma|(I-R_{\bar{E}}^2)\mu_x|^2 &\to 0, \\
     - 2k_{Rp}  | (I- R_{\bar{E}}^2)\mu_Z |^2 &\to 0, \\
     - 2k_{m}  | (I- R_{\bar{E}}^2) \mathring{y}_m |^2 &\to 0.
\end{align*}
As a result $(I- R_{\bar{E}}^2)\mu_Z \to 0$ and $(I- R_{\bar{E}}^2)\mathring{y}_m \to 0$.
Using Lemma 5 in \cite{van2023vaa}, it follows that $R_{\bar{E}}^2 \to I_3$ or $R_{\bar{E}} \to R_{\bar{E}}^\top$.
Hence, $R_{\bar{E}} \to I_3$ or $R_{\bar{E}}\to U\Lambda U^\top$, where $U \in \SO$ and $\Lambda = \diag(1, -1, -1)$.
In the first case when $R_{\bar{E}} \to I_3$, the total error $\bar{E} \to I_{n+5}$ asymptotically.
Linearising the rotation error dynamics about $R_{\bar{E}} \approx I_3 + \varepsilon_R^\times$, for $\varepsilon_R \in \mathbb{R}^3$, and $V_{\bar{E}} \approx 0$, we have
\begin{align*}
    \dot{\varepsilon}_R =  A(t)\varepsilon_R, 
\end{align*}
where $A(t) = 4(k_{Rp}\sigma \mu_x^\times \mu_x^\times + k_{Rx}\mu_Z^\times \mu_Z^\times + k_m \mathring{y}_m^\times \mathring{y}_m^\times)$ is negative semi-definite and persistently exciting.
Hence, the local exponential stability of the rotation error dynamics follows from \cite[Theorem 1]{morgan1977satbility}.
In the second case, $\bar{E} \to (U\Lambda U^\top, 0)\in \mathcal{E}_u$ asymptotically, which can be shown as the set of unstable equilibria \cite[Theorem 4.2, item (ii)]{van2023vaa}.
\underline{Proof of item (iv)}: As the rotation error $R_{\bar{E}}\to I_3$, the estimated attitude $\hat{R} \to R$.
As $V_{\bar{E}}\to 0$ and $R_{\bar{E}}\to I_3$, we have $\hat{V}A_Z = R_{\bar{E}}^\top VA_Z - R_{\bar{E}}^\top V_{\bar{E}} - (R_{\bar{E}}^\top-I_3)V_Z \to VA_Z$.
Finally, $\hat{V} \to V$ due to invertibility of $A_Z$.
Therefore, the estimated state converges to the true state, $\hat{X} \to X$, as long as the error converges to identity, $\bar{E} \to I_{n+5}$.  

\end{proof}
\section{Simulations}\label{sec:simulations}
To validate the proposed observer design, we simulate a robot flying in a circular trajectory of radius 1m at a height of 1m, with a velocity of 1m/s.
The initial robot states are
\begin{align*}
    R(0) = I_3, && v(0) = \mathbf{e}_2, && x(0) = \mathbf{e}_1 + \mathbf{e}_3.
\end{align*}
The angular velocity $\Omega$ and acceleration $a$ measured by the IMU are
\begin{align*}
    \Omega = \begin{pmatrix}
        0 & 0 & 1
    \end{pmatrix}^\top, && a = -\mathbf{e}_1 - g\mathbf{e}_3.
\end{align*}
The robot measures five landmarks on the ground plane, with positions given by $p_1=(0.5\ 0.5\ 0)^\top,\ p_2=(0.5\ -0.5\ 0)^\top,\ p_3 = (-1\ 0.5\ 0)^\top,\ p_4=(1\ 1\ 0)^\top$ and $p_5=(-1.2\ -1.2\ 0)^\top$.
We initialised the observer states as
\begin{align*}
    &\hat{R}(0)=\exp(0.25\pi \mathbf{b}^\times), && \hat{v}(0) = (0\ 0\ 0)^\top, \\
    &\hat{x}(0) = (0\ 0\ 0)^\top, && \hat{p}_i(0) =(0\ 0\ 0)^\top,
\end{align*}
where $\mathbf{b} = (1\ 1\ 1)^\top$,
and the auxiliary states as $V_Z(0) = 0_{3 \times (n+2)}$ and 
\begin{align*}
    A_Z(0) &= \begin{pmatrix}
        36.7423 & 0 & 15.8114\ \mathbf{1}_n^\top \\
        -0.2722 & 1.3878 & -3.1623\ \mathbf{1}_n^\top \\
        \mathbf{0}_n & \mathbf{0}_n & 3.1623\ I_n 
    \end{pmatrix}.
\end{align*}
We chose the observer gains
\begin{align*}
    &k_x = 1, && k_p = 2, && q = 0.1, \\
    &k_{Rx} = 0.001, && k_{Rp}=0.0005, && k_m=0.1\ .
\end{align*}
Note that $k_x, k_p, q$ satisfy the condition \eqref{eq:condition_gain}.
The system and observer equations are simulated for 40s using Lie group Euler integration at 2000 Hz.
The GNSS measurements are available periodically for regular intervals of duration 5s starting from 5s.

Figure \ref{fig:traj} shows the evolution of the true and estimated robot and landmark positions, and shows that the estimates converge over time.
Figure \ref{fig:3axis_err} shows that the errors in attitude, velocity, robot position, and landmark positions, all converge to zero over time.
It also shows that the value of the Lyapunov function is strictly decreasing, verifying the results in Theorem \ref{thm:main_result}. 
In particular, it can be seen that the position converge only when the GNSS signal is available (shown by green highlighting).

\begin{figure}[h!]
    \centering
    \includegraphics[trim=0cm 0cm 0cm 0cm, clip, width=0.45\textwidth]{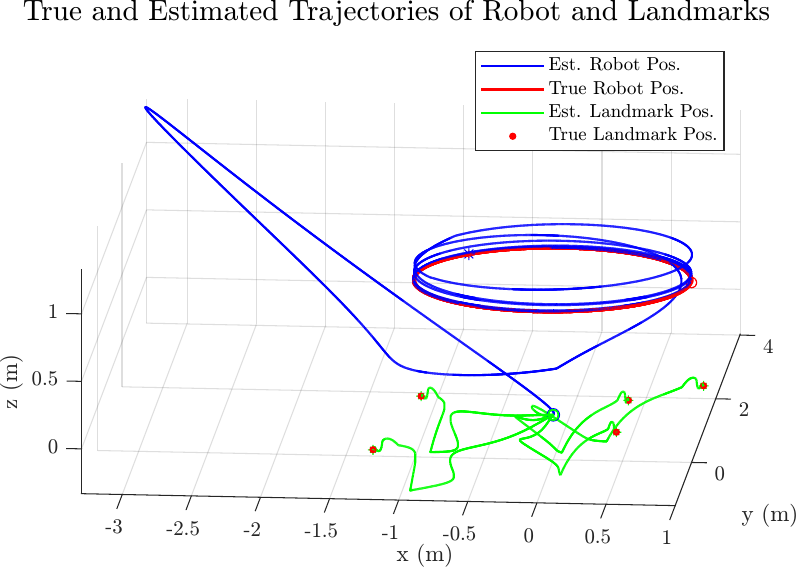}
    \caption{The trajectories of the true and estimated robot and landmark positions over time. 
    All the initial and final positions are marked with $\circ$ and $*$ respectively.
    }
    \label{fig:traj}
\end{figure}
\begin{figure}[h!]
    \centering
    \includegraphics[trim=0cm 0cm 0cm 0cm, clip, width=0.48\textwidth]{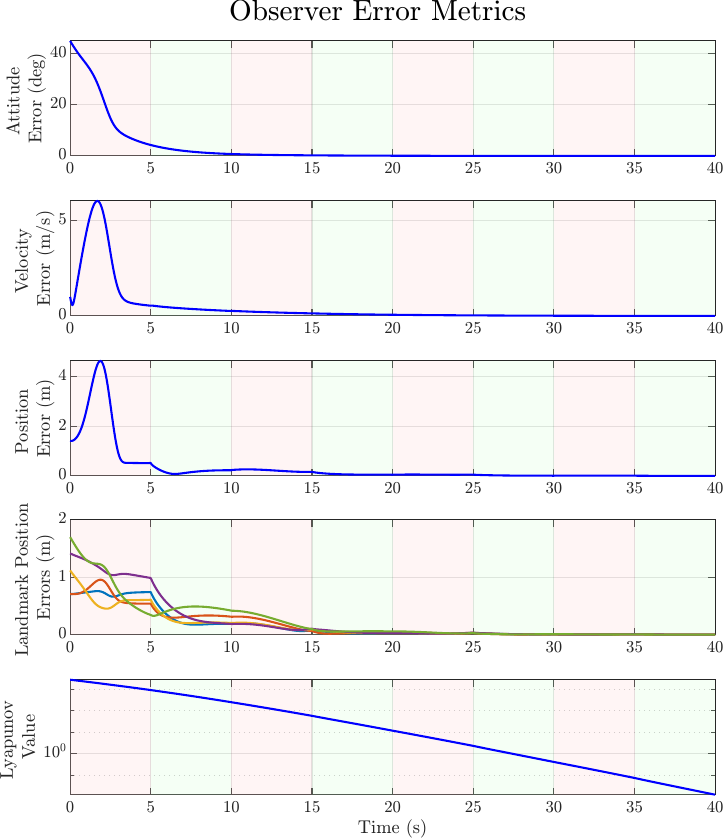}
    \caption{The errors in attitude, velocity and the positions of the robot and landmarks over time. 
    The green and red background shows whether GNSS is available or not.
    }
    \label{fig:3axis_err}
\end{figure}

\section{Conclusion}\label{sec:conclusion}
This paper presented a novel synchronous observer design for Landmark-inertial SLAM aided with GNSS position and magnetometer measurements.
In the standard inertial SLAM problem with only landmark measurements, the robot and landmark positions in the inertial frame and the yaw are not observable, necessitating the addition of GNSS position and magnetometer measurements to make the entire system observable.
The measurement model also of GNSS signals explicitly considered intermittency, which is a significant theoretical and practical challenge to observer design.
We exploited the modularity of synchronous observers to develop correction terms individually before adding them together, greatly simplifying the design process, and resulting in an observer where the intermittent measurements are seamlessly integrated.
We showed the observer error to be almost globally asymptotically and locally exponentially stable as long as the GNSS signal is temporally persistently exciting.
The simulations validate the convergence of states and errors even with poor initial estimates.
This work contributes to the ongoing development of synchronous observers that provide powerful convergence guaranties for inertial navigation and SLAM.

\bibliographystyle{IEEEtran}
\bibliography{IEEEabrv,mybibfile}

\begin{thebibliography}{10}
\providecommand{\url}[1]{#1}
\csname url@rmstyle\endcsname
\providecommand{\newblock}{\relax}
\providecommand{\bibinfo}[2]{#2}
\providecommand\BIBentrySTDinterwordspacing{\spaceskip=0pt\relax}
\providecommand\BIBentryALTinterwordstretchfactor{4}
\providecommand\BIBentryALTinterwordspacing{\spaceskip=\fontdimen2\font plus
\BIBentryALTinterwordstretchfactor\fontdimen3\font minus \fontdimen4\font\relax}
\providecommand\BIBforeignlanguage[2]{{%
\expandafter\ifx\csname l@#1\endcsname\relax
\typeout{** WARNING: IEEEtran.bst: No hyphenation pattern has been}%
\typeout{** loaded for the language `#1'. Using the pattern for}%
\typeout{** the default language instead.}%
\else
\language=\csname l@#1\endcsname
\fi
#2}}

\bibitem{bailey_slam:part2}
H.~Durrant-Whyte and T.~Bailey, ``Simultaneous localization and mapping: part i,'' \emph{IEEE Robotics \& Automation Magazine}, vol.~13, no.~2, pp. 99--110, 2006.

\bibitem{kaess2012isam2}
M.~Kaess, H.~Johannsson, R.~Roberts, V.~Ila, J.~J. Leonard, and F.~Dellaert, ``isam2: Incremental smoothing and mapping using the bayes tree,'' \emph{The International Journal of Robotics Research}, vol.~31, no.~2, pp. 216--235, 2012.

\bibitem{van2021constructive}
P.~van Goor, R.~Mahony, T.~Hamel, and J.~Trumpf, ``Constructive observer design for visual simultaneous localisation and mapping,'' \emph{Automatica}, vol. 132, p. 109803, 2021.

\bibitem{boughellaba2025nonlinear}
M.~Boughellaba, S.~Berkane, and A.~Tayebi, ``Nonlinear observer design for landmark-inertial simultaneous localization and mapping,'' in \emph{2025 European Control Conference (ECC)}.\hskip 1em plus 0.5em minus 0.4em\relax IEEE, 2025, pp. 1967--1972.

\bibitem{lourencco2016simultaneous}
P.~Louren{\c{c}}o, B.~J. Guerreiro, P.~Batista, P.~Oliveira, and C.~Silvestre, ``Simultaneous localization and mapping for aerial vehicles: a 3-d sensor-based gas filter,'' \emph{Autonomous Robots}, vol.~40, pp. 881--902, 2016.

\bibitem{johansen2016globally}
T.~A. Johansen and E.~Brekke, ``Globally exponentially stable kalman filtering for slam with ahrs,'' in \emph{2016 19th International Conference on Information Fusion (FUSION)}.\hskip 1em plus 0.5em minus 0.4em\relax IEEE, 2016, pp. 909--916.

\bibitem{barrau2015ekf}
A.~Barrau and S.~Bonnabel, ``An ekf-slam algorithm with consistency properties,'' \emph{arXiv preprint arXiv:1510.06263}, 2015.

\bibitem{mahony2008nonlinear}
R.~Mahony, T.~Hamel, and J.-M. Pflimlin, ``Nonlinear complementary filters on the special orthogonal group,'' \emph{IEEE Transactions on automatic control}, vol.~53, no.~5, pp. 1203--1218, 2008.

\bibitem{baldwin2009nonlinear}
G.~Baldwin, R.~Mahony, and J.~Trumpf, ``A nonlinear observer for 6 dof pose estimation from inertial and bearing measurements,'' in \emph{2009 IEEE International Conference on Robotics and Automation}.\hskip 1em plus 0.5em minus 0.4em\relax IEEE, 2009, pp. 2237--2242.

\bibitem{mahony2017geometric}
R.~Mahony and T.~Hamel, ``A geometric nonlinear observer for simultaneous localisation and mapping,'' in \emph{2017 IEEE 56th Annual Conference on Decision and Control (CDC)}.\hskip 1em plus 0.5em minus 0.4em\relax IEEE, 2017, pp. 2408--2415.

\bibitem{saha2025synchronous}
A.~Saha, P.~van Goor, A.~Franchi, and R.~Banavar, ``Synchronous observer design for landmark-inertial slam with almost-global convergence,'' \emph{arXiv preprint arXiv:2511.04531}, 2025.

\bibitem{van2023vaa}
P.~van Goor, T.~Hamel, and R.~Mahony, ``Constructive equivariant observer design for inertial velocity-aided attitude,'' \emph{IFAC-PapersOnLine}, vol.~56, no.~1, pp. 349--354, 2023.

\bibitem{VANGOOR2025112328}
P.~{van Goor}, T.~Hamel, and R.~Mahony, ``Synchronous observer design for inertial navigation systems with almost-global convergence,'' \emph{Automatica}, vol. 177, p. 112328, 2025.

\bibitem{martinelli2013observability}
A.~Martinelli \emph{et~al.}, ``Observability properties and deterministic algorithms in visual-inertial structure from motion,'' \emph{Foundations and Trends{\textregistered} in Robotics}, vol.~3, no.~3, pp. 139--209, 2013.

\bibitem{kim20056dof}
J.~Kim and S.~Sukkarieh, ``6dof slam aided gnss/ins navigation in gnss denied and unknown environments,'' \emph{Journal of Global Positioning Systems}, vol.~4, no. 1-2, pp. 120--128, 2005.

\bibitem{boche2025okvis2}
S.~Boche, J.~Jung, S.~B. Laina, and S.~Leutenegger, ``Okvis2-x: Open keyframe-based visual-inertial slam configurable with dense depth or lidar, and gnss,'' \emph{IEEE Transactions on Robotics}, 2025.

\bibitem{2012_lee_IntroductionSmoothManifolds}
J.~M. Lee, \emph{Introduction to {{Smooth Manifolds}}}, 2nd~ed., ser. Graduate {{Texts}} in {{Mathematics}}.\hskip 1em plus 0.5em minus 0.4em\relax Springer, 2012.

\bibitem{van2021autonomous}
P.~van Goor and R.~Mahony, ``Autonomous error and constructive observer design for group affine systems,'' in \emph{2021 60th IEEE Conference on Decision and Control (CDC)}.\hskip 1em plus 0.5em minus 0.4em\relax IEEE, 2021, pp. 4730--4737.

\bibitem{morgan1977satbility}
\BIBentryALTinterwordspacing
A.~P. Morgan and K.~S. Narendra, ``On the uniform asymptotic stability of certain linear nonautonomous differential equations,'' \emph{SIAM Journal on Control and Optimization}, vol.~15, no.~1, pp. 5--24, 1977. [Online]. Available: \url{https://doi.org/10.1137/0315002}
\BIBentrySTDinterwordspacing

\end{thebibliography}

\appendices
\section{Proof of Lemma \ref{lem:landmark_correction}}\label{app:proof_lemma_landmark}

The correction term $W_\Delta^p$ can be simplified
\begin{align*}
    W_\Delta^p &= -(k_p + nk_{Rp})\hat{R}(Y - \hat{Y})C^\top A_Z^{-\top} \\
    &= -(k_p + nk_{Rp})\hat{R}(-R^\top VC + \hat{R}^\top \hat{V}C)C^\top A_Z^{-1} \\
    &= (k_p + nk_{Rp})R_{\bar{E}}^\top (V - R_{\bar{E}}\hat{V})CC^\top A_Z^{-\top}.
\end{align*}
Thus, the term $(I-R_{\bar{E}})W_\Gamma^p - R_{\bar{E}}W_\Delta^p$ can be written as
\begin{align*}
    &(I-R_{\bar{E}})W_\Gamma^p - R_{\bar{E}}W_\Delta^p \\
    =& (k_p + nk_{Rp})(I- R_{\bar{E}})V_ZA_Z^{-1}CC^\top A_Z^{-\top} \\\quad
     &- (k_p + nk_{Rp})(VA_Z- R_{\bar{E}}\hat{V}A_Z)A_Z^{-1}CC^\top A_Z^{-\top} \\
    =& -(k_p + nk_{Rp})((VA_Z - R_{\bar{E}}\hat{V}A_Z) \\ \quad &- (I-R_{\bar{E}})V_Z)A_Z^{-1}CC^\top A_Z^{-\top} \\
    =& -(k_p + nk_{Rp})V_{\bar{E}} A_Z^{-1}CC^\top A_Z^{-\top}.
\end{align*}
Substituting this into the $V_{\bar{E}}$ dynamics yields
\begin{align*}
    \dot{V}_{\bar{E}} &= -V_{\bar{E}} S_\Gamma^p + (I-R_{\bar{E}})W_\Gamma^p - R_{\bar{E}}W_\Delta^p \\
    &= \frac{k_p}{2} V_{\bar{E}} A_Z^{-1} CC^\top A_Z^{-\top} \\ \quad &- (k_p + nk_{Rp})V_{\bar{E}}A_Z^{-1}CC^\top A_Z^{-\top} \\
    & = -(\frac{k_p}{2} + nk_{Rp}) V_{\bar{E}}A_Z^{-1}CC^\top A_Z^{-\top}.
\end{align*}
Hence, the translation error term in the Lyapunov function derivative is
\begin{align*}
    2\langle V_{\bar{E}}, \dot{V}_{\bar{E}} \rangle 
    &= -(k_p + 2nk_{Rp})\langle V_{\bar{E}}, V_{\bar{E}} A_Z^{-1}CC^\top A_Z^{-\top} \rangle \\
    &= -(k_p + 2nk_{Rp}) |V_{\bar{E}} A_Z^{-1}C|^2.
\end{align*}
For simplifying the correction term $\Omega_\Delta^p$, we compute
\begin{align*}
    &\hat{R}(Y-\hat{Y})\mathbf{1}_n 
    \\&=  -R_{\bar{E}}^\top (V-R_{\bar{E}}\hat{V})C \mathbf{1}_n\\
    & = -R_{\bar{E}}^\top (VA_Z - R_{\bar{E}}\hat{V}A_Z)A_Z^{-1}C\mathbf{1}_n \\
    &= -R_{\bar{E}}^\top V_{\bar{E}}A_Z^{-1}C\mathbf{1}_n - (R_{\bar{E}}^\top - I)V_ZA_Z^{-1}C\mathbf{1}_n.
\end{align*}
Now, the correction term $\Omega_\Delta^p$ can be rewritten as
\begin{align*}
    \Omega_\Delta^p =& -4k_{Rp}(V_ZA_ZC\mathbf{1}_n)^\times (R_{\bar{E}}^\top V_{\bar{E}}A_Z^{-1}C\mathbf{1}_n) \\ \quad &- 4k_{Rp} (V_ZA_ZC\mathbf{1}_n)^\times (R_{\bar{E}}^\top V_ZA_Z^{-1}C\mathbf{1}_n).
\end{align*}
The rotational part of the Lyapunov function derivative is therefore given by
\begin{align*}
    & tr(R_{\bar{E}}\Omega_\Delta^\times) \\
    &= 4k_{Rp} tr(R_{\bar{E}}((R_{\bar{E}}^\top V_ZA_Z^{-1}C\mathbf{1}_n)^\times (V_ZA_Z^{-1}C\mathbf{1}_n))^\times) \\
    &\qquad +4k_{Rp} tr(R_{\bar{E}}((R_{\bar{E}}^\top V_{\bar{E}}A_Z^{-1}C\mathbf{1}_n)^\times (V_ZA_Z^{-1}C\mathbf{1}_n) )^\times) \\
    &= -2k_{Rp} |(I- R_{\bar{E}}^2)V_ZA_Z^{-1}C \mathbf{1}_n|^2 \\& \quad -4k_{Rp}\langle V_{\bar{E}}A_Z^{-1}C\mathbf{1}_n, (I- R_{\bar{E}}^2)V_ZA_Z^{-1}C \mathbf{1}_n \rangle
\end{align*}
Finally, the derivative of the Lyapunov function satisfies
\begin{align*}
    \dot{\mathcal{L}}
    &= -(k_p + 2nk_{Rp}) |V_{\bar{E}} A_Z^{-1}C|^2 
    \\ & \qquad
    -2k_{Rp} |(I- R_{\bar{E}}^2)V_ZA_Z^{-1}C \mathbf{1}_n|^2 
    \\ & \qquad
    -4k_{Rp}\langle V_{\bar{E}}A_Z^{-1}C\mathbf{1}_n, (I- R_{\bar{E}}^2)V_ZA_Z^{-1}C \mathbf{1}_n \rangle \\
    & \le -(k_p + 2nk_{Rp}) |V_{\bar{E}} A_Z^{-1}C|^2 
    \\ & \qquad
    -2k_{Rp} |(I- R_{\bar{E}}^2)V_ZA_Z^{-1}C \mathbf{1}_n|^2 
    \\ & \qquad
    + 4k_{Rp} |V_{\bar{E}}A_Z^{-1}C\mathbf{1}_n|| (I- R_{\bar{E}}^2)V_ZA_Z^{-1}C \mathbf{1}_n | \\
    & \le -(k_p + 2nk_{Rp}) |V_{\bar{E}} A_Z^{-1}C|^2 
    \\ & \qquad
    -2k_{Rp} |(I- R_{\bar{E}}^2)V_ZA_Z^{-1}C \mathbf{1}_n|^2 
    \\ & \qquad
    + 4k_{Rp}\sqrt{n} |V_{\bar{E}}A_Z^{-1}C|| (I- R_{\bar{E}}^2)V_ZA_Z^{-1}C \mathbf{1}_n | \\
    & \le -k_p|V_{\bar{E}}A_Z^{-1}C|^2 
    \\ & \qquad- 2k_{Rp}(\sqrt{n} |V_{\bar{E}}A_Z^{-1}C| - | (I- R_{\bar{E}}^2)V_ZA_Z^{-1}C \mathbf{1}_n |)^2.
\end{align*}
Hence, $\dot{\mathcal{L}}$ is negative semi-definite as stated.

\section{Proof of Boundedness of $P$}\label{app:proof_P_bounded}

We will show here that the matrix $P$ defined in Theorem~\ref{thm:main_result} has lower and upper bounded eigenvalues.
Let 
\begin{align*}
    P = \begin{pmatrix}
        s_v & s_{vx} & S_{pv}^\top \\
        s_{vx} & s_{x} & S_{px}^\top \\
        S_{pv} & S_{px} & S_p
    \end{pmatrix} \in \mathbb{R}^{(2+n)\times (2+n)},
\end{align*}
where $s_v,s_{vx}, s_x \in \mathbb{R}$, $S_{pv}, S_{px} \in \mathbb{R}^n$, and $S_p \in \mathbb{R}^{n\times n}$.
The dynamics of all the terms in $P$ can be written as 
\begin{align*}
    \dot{s}_v &= - 2s_{vx} - 2qs_v ,& \dot{S}_{pv} &= -S_{px} - 2q S_{pv},\\
    \dot{s}_{vx} &= -s_x - 2qs_{vx}, & \dot{S}_{px} &= -k_p \mathbf{1}_n - 2q S_{px},\\
    \dot{s}_x &= (k_x\sigma + k_p n) - 2q s_x, & \dot{S}_{p} &= k_p I_n - 2q S_{p}.
\end{align*}

We first find upper and lower bounds on $s_x, s_v, s_{vx}$.
The lower bound on $s_x$ is given by
    \begin{align}
    s_x(t) &= e^{-2qt}s_x(0) + \frac{nk_p}{2q}(1-e^{-2qt}) \notag \\&\qquad 
    + k_x\int_0^t\sigma(s)e^{-2q(t-s)}ds \notag \\
    &\ge \frac{nk_p}{2q} + k_x \int_{t-T}^t \sigma(s) e^{-2q(t-s)}ds \notag \\
    &\ge \frac{nk_p}{2q} + k_x e^{-2qT}\int_{t-T}^t \sigma(s) ds \notag \\
    &\ge \frac{nk_p}{2q} + k_x e^{-2qT}\tau, \label{eq:use_tpe}
\end{align}
where \eqref{eq:use_tpe} follows from TPE of $\sigma$.
Also, $s_x(t)$ is upper-bounded when $\sigma(t) \equiv 1$, that is
\begin{align*}
    \frac{nk_p}{2q} +\delta \le s_x(t) \le \frac{nk_p+k_x}{2q},
\end{align*}
where $\delta := k_x e^{-2qT}\tau$.
Consequently, upper and lower bounds can be derived for $s_{vx}$ and $s_v$,
\begin{align*}
    -\frac{nk_p+k_x}{4q^2} \le &s_{vx}(t) \le -\frac{nk_p}{4q^2} -\frac{\delta}{2q}, \\
    \frac{nk_p}{4q^3} +\frac{\delta}{2q^2}\le &s_v(t) \le \frac{nk_p+k_x}{4q^3}.
\end{align*}

The terms $S_{p}, S_{px}$ and $S_{pv}$ are all stable and time-invariant linear ODEs that have steady solutions
\begin{align*}
    S_p = \frac{k_p}{2q} I_n, &&
    S_{px} = -\frac{k_p}{2q} \mathbf{1}_n, &&
    S_{pv} = \frac{k_p}{4q^2}\mathbf{1}_n.
\end{align*}
These are exactly the choices made in the statement of Theorem~\ref{thm:main_result}.
Since the entries of $P$ are all bounded above, it must be that the eigenvalues of $P$ are also bounded above.

To see that the eigenvalues of $P$ are bounded below, it now suffices to show that the determinant of $P$ is bounded below also.
We write  $P$ takes in block matrix form as
\begin{align*}
    &P = \begin{pmatrix}
        A_p & B_p\\ B_p^\top & C_p
    \end{pmatrix}, \quad \text{where}\\
    &A_p = \begin{pmatrix}
        s_v & s_{vx} \\s_{vx} & s_{x}
    \end{pmatrix}, \ B_p = \begin{pmatrix}
        \frac{k_p}{4q^2}\mathbf{1}_n^\top \\ -\frac{k_p}{2q} \mathbf{1}_n^\top
    \end{pmatrix},\ C_p = \frac{k_p}{2q} I_n.
\end{align*}
Since the determinant of $C_p$ is a fixed constant, then the determinant of $P$ is bounded below if and only if the determinant of its Schur complement $P/C_p = A_p - B_pC_p^{-1}B_p^\top$ is bounded below.
We have that
\begin{align*}
    P/C_p = \begin{pmatrix}
        s_v - \frac{nk_p}{8q^3} & s_{vx}+\frac{nk_p}{4q^2} \\s_{vx} +\frac{nk_p}{4q^2} & s_{x} - \frac{nk_p}{2q}
    \end{pmatrix},
\end{align*}
and thus the determinant of $P/C_P$ is 
\begin{align*}
    \det{(P/C_P)} &= \left(s_v - \frac{nk_p}{8q^3} \right)\left(s_x - \frac{nk_p}{2q} \right) - \left(s_{vx} - \frac{nk_p}{4q^2} \right)^2 \\
    &\ge \delta \left (\frac{nk_p}{8q^3} + \frac{\delta}{2q^2} \right) - \frac{k_x^2}{16q^4} \\
    & = \frac{2nk_pq\delta + 8q^2\delta^2 - k_x^2}{16q^4} \\
    &= \frac{k_x(2n\tau q e^{-2qT}k_p + (8q^2\tau^2 e^{-4qT}-1)k_x)}{16q^4}.
\end{align*}
Since this term is constant and designed to be greater than zero by our selection of gains in Theorem~\ref{thm:main_result}, this completes the proof that the eigenvalues of $P$ are also bounded below.

\section{Dynamics of $V_Z$}\label{app:V_z_dynamics}
Direct computation of the dynamics of $V_Z$ yields
\begin{align*}
    \dot{V}_Z &=  W_G A_Z - W_\Gamma - V_Z S_\Gamma \\
    &= W_G A_Z + (k_x + k_{Rx})(y_x - \sigma V_ZA_Z^{-1}C_x)C_x^\top A_Z^{-\top}
    \\ &\hspace{1cm}
    - (k_p + nk_{Rp})V_ZA_Z^{-1}CC^\top A_Z^{-\top} 
    \\ &\hspace{1cm}
    +\frac{k_x\sigma}{2} V_ZA_Z^{-1}C_xC_x^\top A_Z^{-\top}
    \\ &\hspace{1cm}
    +\frac{k_p}{2} V_ZA_Z^{-1}CC^\top A_Z^{-\top} - qV_Z \\
    &= -V_Z \bigg ( \sigma \left(\frac{k_x}{2} + k_{Rx}\right) A_Z^{-1}C_xC_x^\top A_Z^{-\top}
    \\ &\hspace{1cm}
    + \left(\frac{k_p}{2} + n k_{Rp}\right) A_Z^{-1}CC^\top A_Z^{-\top} + qI \bigg )
    \\ &\hspace{1cm}
    +W_GA_Z + (k_x + k_{Rx})y_x C_x^\top A_Z^{-\top} \\
    &= -V_Z M + B,
\end{align*}
where 
\begin{align*}
    M &= \bigg ( \sigma \left(\frac{k_x}{2} + k_{Rx}\right) A_Z^{-1}C_xC_x^\top A_Z^{-\top} 
    \\ &\hspace{1cm}
    + \left(\frac{k_p}{2} + n k_{Rp}\right) A_Z^{-1}CC^\top A_Z^{-\top} + qI \bigg ), \\
    B &= W_GA_Z + (k_x + k_{Rx})y_x C_x^\top A_Z^{-\top}.
\end{align*}
\end{document}